\documentclass[11pt,reqno]{amsart}
\usepackage[foot]{amsaddr}

  \usepackage{cite}
  \usepackage[T1]{fontenc}        

\usepackage{orcidlink}
\usepackage{amsmath}
\usepackage{amsthm}
\usepackage{amssymb}
\usepackage{amsfonts}
\usepackage{eucal}
\usepackage{array}
\usepackage{enumerate}
\usepackage{url}
\usepackage{tikz}
\usetikzlibrary {arrows.meta,automata,positioning}

\usepackage{enumitem}
\setlist[enumerate]{label*=(\arabic*)}

\usepackage[capitalise]{cleveref}
\Crefname{enumi}{}{}
\crefname{enumi}{}{}

\usepackage[all]{xy}

\DeclareMathOperator{\MCA}{\mA}                               
\DeclareMathOperator{\MBZ}{\mathbb{Z}}                               
\DeclareMathOperator{\Pow}{\mathcal{P}}                               
\DeclareMathOperator{\NN}{In}                               

\DeclareMathOperator{\dd}{deg}                               
\DeclareMathOperator{\defect}{dfc}                               
\DeclareMathOperator{\Replace}{Repl}

\newtheorem{theorem}{Theorem}

\newtheorem{lemma}[theorem]{Lemma}

\newcommand{\sa}{synchronizing automata}
\newcommand{\san}{synchronizing automaton}
\newcommand{\cra}{completely reachable automata}
\newcommand{\cran}{completely reachable automaton}

\newcommand{\scn}{strongly connected}

\usepackage{multirow}
\usepackage{rotating}

\usepackage{mathptmx}      
\DeclareSymbolFont{rsfscript}{OMS}{rsfs}{m}{n}
\DeclareSymbolFontAlphabet{\mathrsfs}{rsfscript}

\newcommand{\mA}{\mathrsfs{A}}

\newcommand{\mC}{\mathrsfs{C}}

\begin{document}

\title{Completely Reachable Road Coloring}
\author{Mikhail V. Volkov\orcidlink{0000-0002-9327-243X} \and Yinfeng Zhu\orcidlink{0000-0003-1724-5250}}

\address{\normalfont Institute of Natural Sciences and Mathematics, Ural Federal University, Ekaterinburg, Russia}
\email{m.v.volkov@urfu.ru, yinfeng.zhu7@gmail.com}

\begin{abstract}
We characterize the digraphs that admit an edge labeling by letters from a finite alphabet such that the resulting labeled digraph is a completely reachable automaton. This class of digraphs can be recognized in polynomial time. In contrast, we show that, for every fixed alphabet size, the problem of deciding whether a digraph admits an edge labeling with the same property is NP-complete. We also classify the digraphs for which every edge labeling results in a completely reachable automaton.

\end{abstract}

\maketitle

\section*{Introduction}

\emph{Completely reachable automata} are complete deterministic finite automata in which every nonempty subset of the state set occurs as the image of the whole state set under the action of some input word. Such automata arose in studies of the descriptional complexity of formal languages, first implicitly in~\cite{Maslennikova:2012} and later explicitly in~\cite{BV16}.

In~\cite{BV16}, Bondar and the first-named author posed three problems concerning \cra. The first was the problem of efficiently recognizing such automata: is there a polynomial-time algorithm that, given an automaton $\mA$, determines whether $\mA$ is completely reachable? This question was recently answered in the affirmative by Ferens and Szyku\l{}a~\cite{FS23,FS26}. The other two problems were inspired by well-known questions in the theory of \sa, of which \cra{} form a special class.

In connection with the \v{C}ern\'y conjecture (cf. \cite[Sect.~3]{Vo22}), it was suggested in~\cite{BV16} to investigate the following question: if $\mA$ is a \cran{} with $n$ states, how can the length of a shortest word with an $m$-element image ($1\le m<n$) be upper-bounded in terms of $n$ and $m$? Don's conjecture~\cite{Don16}, according to which $n(n-m)$ is such an upper bound, was refuted by the second-named author~\cite{Zhu24}; on the other hand, Ferens and Szyku\l{}a~\cite{FS26} obtained the upper bound
$2n(n-m)-n(1+\frac12+\cdots+\frac1{n-m})$.

Motivated by the Road Coloring Conjecture of Adler, Goodwyn, and Weiss~\cite{AGW77}, subsequently proved by Trahtman~\cite{Tra09}, it was asked in~\cite[Sect.~5]{BV16} which digraphs admit an edge labeling by letters from a finite alphabet such that the resulting labeled digraph becomes a \cran. This is the question we address, thereby completing the research program proposed in~\cite{BV16}.

We assume that the reader is familiar with basic concepts of computational complexity such as polynomial-time reduction and NP-completeness.

\section{Background and overview}
\label{sec:background}

\subsection{Automata }All automata in this paper are complete, deterministic, and finite. Such an automaton (a \emph{DFA}) is a pair $(Q,\Sigma)$ of finite nonempty sets equipped with a map $Q\times\Sigma\to Q$. The elements of $Q$ and $\Sigma$ are referred to as \emph{states} and \emph{letters}, respectively. The map $Q\times\Sigma\to Q$ is called the \emph{action} of letters on states. For a state $q\in Q$ and a letter $a\in\Sigma$, the result of the action of $a$ on $q$ is denoted by $q{\cdot}a$.

A \emph{word} over $\Sigma$ is a finite sequence of letters. We allow the \emph{empty word}, denoted by $\varepsilon$, which is the empty sequence. The set of all words over $\Sigma$ is denoted by $\Sigma^*$. For $q\in Q$ and $\sigma\in\Sigma^*$, the action of $\sigma$ on $q$ in  $(Q,\Sigma)$ is defined recursively by
\[
q{\cdot}w:=\begin{cases}
               q & \text{if } \sigma=\varepsilon,\\
              (q{\cdot}\sigma'){\cdot}a & \text{if } \sigma=\sigma'a \text{ for some } \sigma'\in\Sigma^* \text{ and } a\in\Sigma,
             \end{cases}
\]
and the action of $\sigma$ on a nonempty subset $P\subseteq Q$ is defined by
\[
P{\cdot}\sigma:=\{p{\cdot}\sigma : p\in P\}.
\]

A subset $P \subseteq Q$ is called \emph{reachable} in $(Q,\Sigma)$ if there exists a word $\sigma\in\Sigma^*$ such that $P = Q{\cdot}\sigma$. A DFA is \emph{completely reachable} if every nonempty subset of its state set is reachable. A DFA $(Q,\Sigma)$ is \emph{synchronizing} if some singleton set is reachable in $(Q,\Sigma)$. Obviously, every \cran{} is synchronizing.

\subsection{Digraphs and colorings}  A \emph{digraph} is a quadruple $(V,E,s,t)$, where $V$ and $E$ are sets and $s,t\colon E\to V$ are functions. The elements of $V$ and $E$ are called \emph{vertices} and \emph{edges}, respectively. For an edge $e\in E$, the vertices $s(e)$ and $t(e)$ are called the \emph{source} and the \emph{target} of $e$, respectively, and we say that $e$ is an \emph{edge from $s(e)$ to $t(e)$}. Note that, under this definition, for any $u,v\in V$, there may be several edges from $u$ to $v$; we call such edges \emph{parallel}.

A \emph{path of length} $\ell$ from $u_0$ to $u_\ell$ is an alternating sequence $u_0,e_1,u_1,e_2,\dots,u_{\ell-1}$, $e_\ell,u_\ell$
of vertices $u_0,\dots,u_\ell$ and edges $e_1,\dots,e_\ell$ in which $s(e_i)=u_{i-1}$ and $t(e_i)=u_i$ for each $i=1,\dots,\ell$. A digraph is \emph{\scn} if, for every pair of distinct vertices $u,v$, there exists a path from $u$ to $v$.

A \emph{labeling} of a digraph $(V,E)$ with label set $\Lambda$ is a map $E\to\Lambda$.

For a set $\Sigma$, let $\Pow(\Sigma)$ denote the set of all nonempty subsets of $\Sigma$. An arbitrary DFA $(Q,\Sigma)$ can be represented as a labeled digraph with vertex set $Q$, edge set
\[
E:=\{(p,q)\in Q\times Q : (\exists\, a\in\Sigma)\, p{\cdot}a=q \},
\]
the functions $s,t\colon E\to V$ defined by $s(p,q):=p$, $t(p,q)=q$, and labeling with label set $\Pow(\Sigma)$ defined by $(p,q)\mapsto\{a\in\Sigma: p{\cdot}a=q \}$. For an illustration, Fig.~\ref{fig:C4} shows the labeled digraph of the \v{C}ern\'{y} automaton $\mC_4$ \cite[Fig.~3]{Cerny:1964}, which has states $0,1,2,3$, letters $a,b$, and the action of letters on states defined by
\[
0{\cdot}a=0{\cdot}b:=1,\ \ m{\cdot}a:=m,\ \ m{\cdot}b:=m+1\mkern-12mu\pmod 4 \ \text{for }  m=1,2,3.
\]
\begin{figure}[ht]
\centering
\begin{tikzpicture}[scale=1,
    x=1.2mm, y=1.2mm,
    >=stealth,
    auto,
    state/.style={circle, draw, inner sep=1.5pt, minimum size=6mm},
    every edge/.style={draw, ->,thick}
]
    \node[state] (A) at (0,18) {0};
    \node[state] (B) at (18,18) {1};
    \node[state] (C) at (18,0) {2};
    \node[state] (D) at (0,0) {3};

   \path (A) edge node {$a,b$} (B)
          (B) edge node {$b$}   (C)
          (C) edge node {$b$}   (D)
          (D) edge node {$b$}   (A);

\draw[->,thick] (B) to[out=15,in=75,looseness=8]
  node[above right] {$a$} (B);

\draw[->,thick] (C) to[out=-75,in=-15,looseness=8]
  node[below right] {$a$} (C);

\draw[->,thick] (D) to[out=-105,in=-165,looseness=8]
  node[below left] {$a$} (D);
\end{tikzpicture}
\vspace{-6mm}
\caption{The digraph of the automaton $\mathrsfs{C}_4$}
\label{fig:C4}
\end{figure}
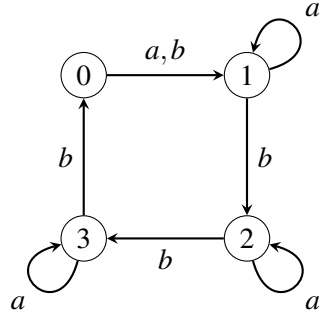
(In Fig.~\ref{fig:C4} and below we omit \{ \} in the edge labels.)

The digraph obtained from the labeled digraph of a DFA $\mA$ by omitting the labels is called the \emph{underlying digraph} of $\mA$. We shall call a DFA \emph{\scn} whenever its underlying digraph is \scn.

Due to completeness and determinism, in the labeled digraph of a DFA $(Q,\Sigma)$, the labels of edges with a common source form a partition of $\Sigma$. This observation admits a converse in the following sense. Suppose that a labeling $\alpha$ of a digraph $G=(V,E,s,t)$ assigns nonempty subsets of a set $\Sigma$ to edges in such a way that, for every vertex $v\in V$, the family $\{\alpha(e): s(e)=v\}$ forms a partition of $\Sigma$; such a labeling is called a \emph{road coloring}. Every road coloring $\alpha$ defines an action of $\Sigma$ on~$V$: for $v\in V$ and $a\in\Sigma$, set $v{\cdot}a:=t(e)$, where $e$ is the unique edge such that $s(e)=v$ and the set $\alpha(e)$ contains~$a$. We thus obtain a DFA, referred to as a~\emph{coloring} of the digraph $G$ and denoted by $\MCA(G,\alpha)$.

\subsection{Road coloring problems and our contribution}

 Given a class $\mathbf{C}$ of automata, one may ask for which digraphs $G$ there exists a road coloring $\alpha$ the DFA $\MCA(G,\alpha)$ belongs to $\mathbf{C}$.  One may also ask for which digraphs $G$ \textbf{every} road coloring $\alpha$ is such that the DFA $\MCA(G,\alpha)$ belongs to $\mathbf{C}$.  We will refer to these questions as the \emph{road coloring problem} (RCP) and the \emph{universal road coloring problem}, respectively, for the class $\mathbf{C}$.

With this terminology, the famous Road Coloring Conjecture of Adler, Goodwyn, and Weiss~\cite{AGW77}, confirmed by Trahtman~\cite{Tra09}, concerns the RCP for the class $\mathbf{SCS}$ of \scn{} \sa. The name of the problem stems from the following property of \scn{} \sa{} which admits a `practical' interpretation. In every \scn\ \san\ $\mA=(Q,\Sigma)$, one can assign to each state $q\in Q$ a word $\sigma_q\in\Sigma^*$ such that, by following the path whose labels spell $\sigma_q$, one is guaranteed to arrive at $q$ from any initial state. Indeed, since $\mA$ is synchronizing, there exist a state $p\in Q$ and a word $\tau\in\Sigma^*$ such that $Q{\cdot}\tau=\{p\}$; since $\mA$ is \scn, there exists a word $\upsilon\in\Sigma^*$ such that $p{\cdot}\upsilon=q$. Then the word $\tau\upsilon$ has the desired property that $r{\cdot}\tau\upsilon=q$ for every $r\in Q$.

Suppose now that a bus carrying tourists arrives in a city and parks at some location. The tourists are then free to explore the city on their own, after which they must return to the bus. If the city's map can be colored into a \san, then every tourist can find their way back to the bus from wherever they happen to be by following the same sequence of colors, determined solely by the location where the bus is parked.

The mathematical motivation for the Road Coloring Conjecture in~\cite{AGW77} came from a question in symbolic dynamics, but in~\cite[Sect.~5]{AGW77}, the authors were able to resolve that question by a different route, bypassing the conjecture altogether. We believe that road coloring problems are interesting \emph{per se}, even without external motivation, as a kind of `reverse engineering' of automata: we aim to understand the extent to which certain important properties of automata are determined by their underlying digraphs.

We consider the RCP for the class $\mathbf{CR}$ of \cra. Although \cra{} form a subclass of \scn{} \sa, our solution to the RCP for \cra{} is not a specialization of Trahtman's. Indeed, a criterion guaranteeing that a coloring of a DFA belongs to a class $\mathbf{C}$ of automata yields, for a subclass $\mathbf{C}'\subset\mathbf{C}$, only a necessary condition for membership in $\mathbf{C}'$. Moreover, it was already observed in~\cite{BV16} that the RCPs for $\mathbf{SCS}$ and $\mathbf{CR}$ differ essentially. Trahtman's result implies that the number of colors plays no role in the case of $\mathbf{SCS}$: if a \scn{} digraph $G$ admits a synchronizing coloring, then it also admits one with the minimum possible number of colors, namely, the maximum out-degree of the vertices of $G$. In contrast,~\cite[Fig.~9]{BV16}, reproduced below as Fig.~\ref{fig:example}, presents a digraph that has no completely reachable coloring with two letters but does admit such a coloring with three letters.

\begin{figure}[hb]
\centering
\begin{tikzpicture}[scale=0.95,
  >=stealth,
  state/.style={circle,draw,minimum size=7mm,inner sep=0pt}
]

\node[state] (A2) at (0,0) {1};
\node[state] (B2) at (2,0) {2};
\node[state] (C2) at (0,2) {4};
\node[state] (D2) at (2,2) {3};

\draw[->,thick] (A2) to[bend right=20] (B2);
\draw[->,thick] (B2) -- (A2);
\draw[->,thick] (C2) -- (A2);
\draw[->,thick] (B2) -- (D2);
\draw[->,thick] (D2) -- (C2);

\draw[->,thick] (A2) edge[loop left] ();

\node[state] (A3) at (5,0) {1};
\node[state] (B3) at (7,0) {2};
\node[state] (C3) at (5,2) {4};
\node[state] (D3) at (7,2) {3};

\draw[->,thick]
(A3) to[bend right=20]
node[midway,below] {$b,c$}
(B3);

\draw[->,thick]
(B3) --
node[midway,above] {$c$}
(A3);

\draw[->,thick]
(C3) --
node[midway,left] {$a,b,c$}
(A3);

\draw[->,thick]
(B3) --
node[midway,right] {$a,b$}
(D3);

\draw[->,thick]
(D3) --
node[midway,above] {$a,b,c$}
(C3);

\draw[->,thick]
(A3) edge[loop left] node {$a$} ();
\end{tikzpicture}

\caption{The digraph on the left has no completely reachable coloring with two letters but admits a~completely reachable coloring with three letters shown on the right.}
\label{fig:example}
\end{figure}
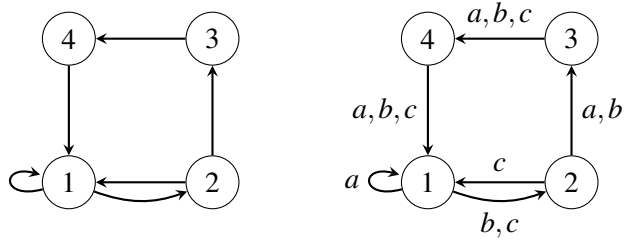

We give a combinatorial characterization of digraphs admitting a~completely reachable coloring in Theorem~\ref{thm:CRC}; this characterization can be verified in polynomial time. On the other hand, Theorem~\ref{thm:NPC} shows that, for every fixed $k\ge 2$, the problem of deciding whether a given digraph admits a completely reachable coloring with $k$ colors is NP-complete. Our final result, Theorem~\ref{thm:ACR}, solves the universal road coloring problem for the class $\mathbf{CR}$. To the best of our knowledge, this problem remains open for the class of \sa---no combinatorial characterization is known for \emph{totally synchronizing} digraphs, that is, digraphs all of whose colorings are synchronizing.

\subsection{Simplification and simple digraphs}
\label{subsec:simple}
Our definition of a digraph allows parallel edges. However, in some considerations related to completely reachable colorings, one can restrict attention to digraphs without parallel edges using the following immediate reduction.

The \emph{simplification} of a digraph $G=(V,E,s,t)$ is the digraph $\bar{G}:=\bigl(V,E/{\sim},\bar s,\bar t\bigr)$, where the equivalence $\sim$ on $E$ is defined by
\[
e\sim f \quad\Longleftrightarrow\quad s(e)=s(f)\ \text{and}\ t(e)=t(f),
\]
and $\bar s([e]):=s(e)$, $\bar t([e])=t(e)$ for every equivalence class $[e]\in E/{\sim}$. A digraph $G$ is \emph{simple} if $\bar G=G$, that is, if $\sim$ is the equality relation.

If $\alpha\colon E\to\Pow(\Sigma)$ is a road coloring of $G=(V,E,s,t)$, then the map $\bar\alpha$ defined by $\bar{\alpha}([e]):=\bigcup_{f\in [e]}\alpha(f)$ is a road coloring of  $\bar{G}$, and the automata $\MCA(G,\alpha)$ and $\MCA(\bar{G},\bar\alpha)$ coincide as they have the same transition function.

Conversely, if $\beta\colon E/{\sim}\,\to\Pow(\Sigma)$ is a road coloring of the digraph $\bar{G}$ and $m:=\max\{|[e]|: e\in E\}$, one can extend $\beta$ to a road coloring $\alpha$ of $G$ with the label set $\Pow(\Theta)$, where $\Theta:=\Sigma_1\cup\dots\cup\Sigma_m$ is the union of $m$ disjoint copies of $\Sigma$, as follows. For every equivalence class $[e]\in E/{\sim}$, fix some order $e_1,\dots,e_k$, $k\le m$, of the edges in $[e]$, and let
\[
\alpha(e_i):=\bigl(\beta([e])\bigr)_i\ \text{ for $i<k$, and }\ \alpha(e_k):=\bigcup_{k\le i\le m}(\beta([e])\bigr)_i.
\]
Here $\bigl(\beta([e])\bigr)_i$ stands for the copy of the subset $\beta([e])\subseteq\Sigma$ in the set $\Sigma_i$. The automata $\MCA(\bar{G},\beta)=(V,\Sigma)$ and $\MCA(G,\alpha)=(V,\Theta)$ are essentially the same; both have $\bar{G}$ as their underlying digraph, and each letter $a\in\Sigma$ has $m$ copies in $\Theta$, each acting the same as $a$. In particular, the DFA $(V,\Sigma)$ is completely reachable if and only if so is $(V,\Theta)$.

We conclude that a digraph admits a completely reachable coloring if and only if its simplification does. Similarly, every coloring of a digraph is completely reachable if and only if every coloring of its simplification is.

\section{Colorings with unrestricted number of colors}
\label{sec:CRC}

A path $u_0,e_1,u_1,\dots,e_\ell,u_\ell$ in a digraph $G=(V,E,s,t)$ is called a \emph{cycle} if $u_\ell=u_0$. The \emph{period} of $G$ is the greatest common divisor of the lengths of all cycles in $G$. A digraph is \emph{aperiodic} if its period is $1$. For a vertex $u\in V$, an \emph{in-neighbor} of $u$ is any vertex $v$ such that $G$ has an edge from $v$ to $u$. For a subset $T\subseteq V$, we denote by $\NN_G(T)$ the set of all in-neighbors of the vertices in $T$. A subset $T\subseteq V$ is \emph{absorbing} if $|T|\le|\NN_G(T)|$.

\begin{theorem} \label{thm:CRC}
A digraph $G=(V,E)$ admits a completely reachable coloring if and only if $G$ is strongly connected, aperiodic, and every subset of\/ $V$ is absorbing. These conditions can be verified in time polynomial in\/ $|V|$, $|E|$.
\end{theorem}

\begin{proof} \emph{Necessity}. Let $\MCA(G,\alpha)=(V,\Sigma)$ be a completely reachable coloring of $G$.

For any two vertices $p,q\in V$, if $\sigma\in\Sigma^*$ satisfies $V{\cdot}\sigma=\{q\}$, then $p{\cdot}\sigma=q$. Therefore, the path in $G$ whose labels in $\MCA(G,\alpha)$ spell $\sigma$ connects $p$ to $q$. Hence, $G$ is strongly connected. Moreover, since the one-element subset $\{q\}$ is reachable, the DFA $\MCA(G,\alpha)$ is synchronizing. That the underlying digraph of any \scn{} \san{} is aperiodic was already shown in Laemmel's report \cite{Laem63} and rediscovered in~\cite{AGW77}. Hence, $G$ is aperiodic.

It remains to verify that every subset $T\subseteq V$ is absorbing. This is clear if $T$ is empty. As $G$ is the underlying digraph of a DFA, each vertex of $G$ is the source of an edge. Hence, $\NN_G(V)=V$, and $V$ is absorbing. Now let $T$ be a proper nonempty subset. Since the DFA $\MCA(G,\alpha)$ is completely reachable, $T=V{\cdot}\sigma$ for some word $\sigma$ that can be written as $\sigma=\sigma'a$ with $\sigma'\in\Sigma^*$ and $a\in\Sigma$. Let  $S:=V{\cdot}\sigma'$. Then $S{\cdot}a=T$, which means that every vertex $t\in T$ is the target of an edge labeled $a$ in $\MCA(G,\alpha)$ whose source lies in $S$. Hence $S\subseteq \NN_G(T)$, and since $|T|\le|S|$, the set $T$ is absorbing.

\smallskip

\emph{Complexity}. Algorithms for verifying the strong connectivity and finding the period of a digraph $G=(V,E,s,t)$ in time $O(|V|+|E|)$ are well known; see, for example, \cite{BJG09}, Sects. 5.2 and 17.8, respectively. Using these algorithms, one can check in time $O(|V|+|E|)$ whether $G$ is strongly connected and aperiodic.

To check whether every subset of $V$ is absorbing, we use the \emph{bipartite representation} $BP(G)$ of $G$. This is the bipartite graph $(V',V'',L)$ whose two parts $V':=\{v': v\in V\}$ and $V'':=\{v'': v\in V\}$ are copies of $V$ and whose link set is
\[
L:=\bigl\{(u',w'')\in V'\times V'' : (\exists\, e\in E)\, s(e)=w \land t(e)=u\bigr\};
\]
see Fig.~\ref{fig:bipartite} for an illustration. (We refer to the elements of $L$ as \emph{links} so that the word `edge' unambiguously refers to an edge of a digraph.)

\begin{figure}[ht]
\centering
\begin{tikzpicture}[scale=0.99,
  >=stealth,
  state/.style={circle,draw,minimum size=7mm,inner sep=0pt}
]

\node[state] (A) at (0,0) {1};
\node[state] (B) at (2.5,0) {2};
\node[state] (C) at (0,2.5) {4};
\node[state] (D) at (2.5,2.5) {3};

\draw[->,thick] (A) to[bend right=20] (B);
\draw[->,thick] (B) -- (A);
\draw[->,thick] (C) -- (A);
\draw[->,thick] (B) -- (D);
\draw[->,thick] (D) -- (C);
\draw[->,thick] (A) edge[loop left] ();

\begin{scope}[xshift=5cm]

\node[state] (A1) at (0,3) {$1'$};
\node[state] (A2) at (0,2) {$2'$};
\node[state] (A3) at (0,1) {$3'$};
\node[state] (A4) at (0,0) {$4'$};

\node[state] (B1) at (4,3) {$1''$};
\node[state] (B2) at (4,2) {$2''$};
\node[state] (B3) at (4,1) {$3''$};
\node[state] (B4) at (4,0) {$4''$};

\draw[thick] (A2) -- (B1); 
\draw[thick] (A1) -- (B2); 
\draw[thick] (A1) -- (B4); 
\draw[thick] (A3) -- (B2); 
\draw[thick] (A4) -- (B3); 
\draw[thick] (A1) -- (B1); 

\end{scope}
\end{tikzpicture}
\caption{A digraph and its bipartite representation}
\label{fig:bipartite}
\end{figure}
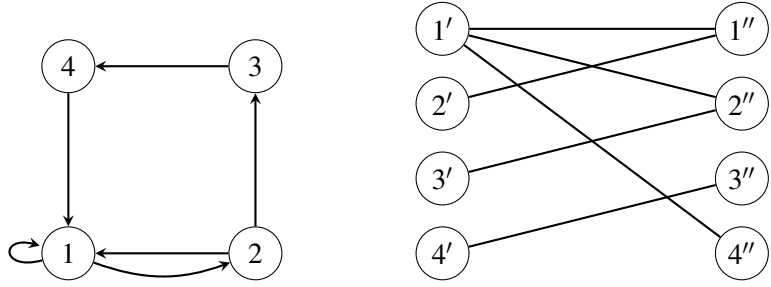
For a subset $T\subseteq V$, we denote by $T'$ and $T''$ its copies in the parts $V'$ and $V''$, respectively, and for a subset $U\subseteq V'$, we denote by $N(U)$ the set of vertices in $V''$ adjacent to at least one vertex in $U$. By construction, $N(T')=\NN_G(T)''$ for every subset $T\subseteq V$. Therefore $T$ is absorbing if and only if $|T'|\le|N(T')|$. Hence, every subset of $V$ is absorbing if and only if $|U|\le|N(U)|$ for every $U\subseteq V'$. By Hall's marriage theorem (see \cite[Theorem~4.11.3]{BJG09}), this is precisely the condition for the graph $BP(G)$ to admit a matching that covers $V'$. (A \emph{matching} in a bipartite graph $B$ is a set of links without common vertices; a matching $M$ \emph{covers} a part of $B$ if each vertex of this part appears in a link from $M$.) The existence of such a matching can be decided in $O(\sqrt{|V|}|E|)$ time using the Hopcroft--Karp algorithm \cite[Sect.~25.1]{Cormen}.

\smallskip

\emph{Sufficiency}. As explained in Sect.~\ref{subsec:simple}, we may assume that $G$ is simple. We again use the bipartite graph $BP(G)$. For each nonempty subset $W\subseteq V$, consider the induced subgraph $H$ of $BP(G)$ with vertex set $W'\cup\NN_G(W)''$. Since every subset of $V$ is absorbing, we have $|U|\le|N(U)|$ for every $U\subseteq W'$. Clearly, $N(U)\subseteq \NN_G(W)''$, and therefore the graph $H$ satisfies the condition of Hall's marriage theorem. Hence, $H$ admits a matching that covers $W'$. Fix such a matching $M$ and define a function $f_W\colon V''\to V'$ as follows:
\begin{itemize}
    \item if $x$ is covered by a link $(x,y)\in M$, then $f_W(x):=y$;
    \item if $x\in \NN_G(W)''$ is not covered by the matching $M$, then $f_W(x)$ is an arbitrary vertex in $W'$ adjacent to $x$;
    \item if $x\in V''\setminus \NN_G(W)''$, then $f_W(x)$ is an arbitrary vertex in $V'$ adjacent to $x$.
\end{itemize}
By construction, $W'=f_W(\NN_G(W)'')$.

Define a labeling $\alpha\colon E\to\Pow(\Sigma)$, where $\Sigma:=\Pow(V)$, by
\[
\alpha(e):=\{U\in\Sigma: f_U(s(e)'')=t(e)'\}.
\]
By construction, for every $v\in V$, the set $\alpha(e)$, where $s(e)=v$, contains $U\in\Sigma$ if and only if $f_U(v'')=t(e)'$. Therefore, our assumption that the digraph $G$ is simple ensures that the family $\{\alpha(e): s(e)=v\}$ forms a partition of $\Sigma$, i.e., that $\alpha$ is a road coloring. We thus obtain the DFA $\MCA(G,\alpha)=(V,\Sigma)$, in which the vertices of $V$ are the states and the nonempty subsets of $V$ are the letters. In $\MCA(G,\alpha)$ we have
\begin{equation}\label{eq:selfmap}
    \NN_G(U){\cdot}U = U
\end{equation}
for every nonempty subset $U\subseteq V$. (On the left-hand side of \eqref{eq:selfmap} $U$ occurs as a letter from $\Sigma$ acting on the set  $\NN_G(U)$; on the right-hand side, $U$ is a set of states.)

By Wielandt's theorem on the minimum positive power of an irreducible nonnegative $n\times n$ matrix \cite[p. 648]{Wi50}, every strongly connected aperiodic digraph with $n$ vertices has a path of length $(n-1)^2+1$ between every pair of vertices. If, for an arbitrary nonempty subset $U_0\subseteq V$, we define $U_\ell:= \NN_G(U_{\ell-1})$ for all $\ell\ge 1$, then $U_\ell$ contains all vertices $w$ for which there exists a path of length $\ell$ between $w$ and a vertex of $U_0$. Hence $U_k = V$, where $k:=(|V|-1)^2+1$. From \eqref{eq:selfmap}, it follows that
\[
V{\cdot}U_{k-1} U_{k-2} \cdots U_1 U_0=U_0,
\]
in particular, the subset $U_0$ is reachable. Hence, the DFA $\MCA$ is completely reachable, and the digraph $G$ admits a completely reachable coloring.
\end{proof}

For a given digraph $G=(V,E,s,t)$, the completely reachable coloring constructed in Theorem~\ref{thm:CRC} uses $2^{|V|-1}$ colors. Of course, this number of colors is far from optimal. However, as the result of the next section shows, the problem of finding a completely reachable coloring with the minimum number of colors is computationally hard.

\section{Colorings with restricted number of colors}
\label{sec:NPC}

\begin{theorem} \label{thm:NPC}
For each $k\ge 2$, determining whether a given digraph admits a completely reachable coloring with $k$ colors is NP-complete.
\end{theorem}

That the problem belongs to NP is easy to see: if one nondeterministically guesses a completely reachable coloring with $k$ colors, Ferens--Szyku\l{}a's algorithm~\cite{FS26} can verify the guess in polynomial time.

For the hardness proof, we use a variant of the classical \textsc{Hamiltonian cycle} problem. We say that a cycle $C$ in a digraph $G = (V,E,s,t)$ is \emph{Hamiltonian} if the length of $C$ is equal to $|V|$ and every vertex of $V$ occurs in $C$. In \cite{Ple79}, Ples\'nik constructed a polynomial-time reduction from the Boolean satisfiability problem \textsc{SAT} to the problem of determining whether a digraph from a special class has a Hamiltonian cycle. We will show that every digraph $G$ from Ples\'nik's class can be transformed in polynomial time into a digraph $\widetilde{G}$ such that $G$ has a Hamiltonian cycle if and only if $\widetilde{G}$ admits a completely reachable coloring with $k\emph{}$ colors. Since \textsc{SAT} is NP-complete, this establishes the NP-hardness of our problem.

We need some notions and notation. For a vertex $v$ of a digraph $G=(V,E,s,t)$, its \emph{out-degree} $\dd_G^+(v)$ is the number of edges $e\in E$ such that $s(e)=v$, and its \emph{in-degree} $\dd_G^-(v)$ is the number of edges $e\in E$ such that $t(e)=v$. A digraph $G=(V,E)$ is $k$-\emph{out-regular} if $\dd_G^+(v)=k$ for all $v\in V$. For a DFA $\MCA = (Q,\Sigma)$, the \emph{defect} of a letter $a \in \Sigma$ is the number $\defect(a):=|Q| - |Q{\cdot}a|$.

\begin{lemma} \label{lem:defect}
Let $G =(V,E,s,t)$ be a $k$-out-regular digraph and $\alpha$ a road coloring of $G$ with color set $\Sigma$ of size $k$. Then in the DFA $\MCA(G, \alpha)$,
\[
   \sum_{a \in \Sigma} \defect(a)\ge   \sum_{v \in V} \max\{0, (k - \dd_G^-(v))\}.
\]
\end{lemma}

\begin{proof}
Since $|\Sigma| = k$ and $G$ is $k$-out-regular, each set $\alpha(e)$ is a singleton, because for every vertex $v$, the sets $\{\alpha(e): s(e)=v\}$  partition  $\Sigma$.

For each $a \in \Sigma$, let $G_a$ be the subdigraph of $G$ whose vertex set is $V$ and whose edge set consists of all edges $e$ with $\alpha(e)=\{a\}$. For every vertex $v$, we have
\[
 \left|  \{a \in \Sigma : \dd_{G_a}^-(v) = 0\} \right|+ \left|  \{a \in \Sigma : \dd_{G_a}^-(v) > 0\} \right| = |\Sigma|=k,
\]
and the summand $\left|  \{a \in \Sigma : \dd_{G_a}^-(v) > 0\} \right|$ does not exceed $\dd_G^-(v)$. Hence
\[
    \left|  \{a \in \Sigma : \dd_{G_a}^-(v) = 0\} \right| \ge \max\{0, (k - \dd_G^-(v))\},
\]
since the left-hand side is nonnegative. Summing up over all $v\in V$ yields
\[
  \sum_{v \in V}  \left|  \{a \in \Sigma : \dd_{G_a}^-(v) = 0\} \right| \ge \sum_{v \in V}\max\{0, (k - \dd_G^-(v))\},
\]
and the sum in the left-hand side is equal to $\sum_{a \in \Sigma}\defect (a)$ since, clearly, $\defect(a)$ equals the number of vertices with in-degree $0$ in $G_a$.
\end{proof}

We now identify a class of digraphs for which Hamiltonicity is equivalent to the existence of a completely reachable coloring with a prescribed number of colors.

\begin{lemma} \label{lem:Hami}
Let $k \ge 2$ and let $G = (V,E,s,t)$ be a $k$-out-regular digraph with a prime number of vertices greater than $\max\{3,k\}$. Suppose that there exists a vertex $x\in V$ such that, for every $v\in V$, there are at least $k-2$ edges from $v$ to $x$, and that deleting $k-2$ edges from $v$ to $x$ from $E$ for each $v\in V$ yields a digraph $G'$ in which
      \begin{itemize}
        \item $\dd_{G'}^-(x) = 1$,
        \item there exists a vertex $y$ such that $\dd_{G'}^-(y) = 3$,
        \item for each $z \in V \setminus \{x,y\}$, $\dd_{G'}^-(z) = 2$.
      \end{itemize}
Then $G$ admits a completely reachable coloring with $k$ colors if and only if $G$ has a Hamiltonian cycle.
\end{lemma}

\begin{proof} \emph{Necessity}. Let $\alpha$ be a completely reachable coloring of $G$ with color set $\Sigma$ of size $k$, and let $n=|V|$. Every $(n-1)$-element subset of $V$ must be reachable in the DFA $\MCA(G,\alpha)=(V,\Sigma)$. This implies that $\Sigma$ contains letters of defect $1$. If $P=V{\cdot}\sigma$ for some word $\sigma=\sigma'c$ with $\sigma'\in\Sigma^*$ and $c\in\Sigma$, then $P\subseteq V{\cdot}c$. Therefore, if $\Sigma$ contained no letters of defect $0$, then words reaching distinct $(n-1)$-element subsets of $V$ would have distinct last letters. Since there are $n$ such subsets and only $k<n$ letters, this is impossible. Hence, $\Sigma$ has a letter of defect $0$. Denote it by $a$. We claim that no further letter has defect 0 and exactly one letter has defect $1$.

Indeed, by Lemma \ref{lem:defect}, we have
    \[
      \sum_{b \in \Sigma}\defect(b) \ge \sum_{v \in V} \max(0, k - \dd_G^-(v)) \ge \sum_{v \in V\setminus\{x\}} \max(0, k - \dd_G^-(v)).
    \]
Conditions imposed on the digraph $G$ imply that $\dd_{G}^-(y) = \dd_{G'}^-(y) = 3$ and $\dd_{G}^-(z) = \dd_{G'}^-(z) = 2$  for each $z \in V \setminus \{x,y\}$. Therefore the sum on the right-hand side equals $(n-2)(k-2)+k-3=(n-1)(k-2)-1$. Now assume that there exist two letters $c,d\ne a$ of defect 0 or 1. Since the defect of any letter does not exceed $n-1$,
\begin{multline*}
     \sum_{b \in \Sigma}\defect(b) = \sum_{b\ne a,c,d}\defect(b) + \defect(a) + \defect (c) + \defect (d)  \le \\
   (n-1)(k-3) + 0 + 1 + 1 = (n-1)(k-2) - (n-3).
\end{multline*}
 Since $n\ge 5$, we have $-1>-(n-3)$, and hence the inequalities
 \[
 \sum_{b \in \Sigma}\defect(b) \ge (n-1)(k-2)-1\  \text{ and }\ \sum_{b \in \Sigma}\defect(b) \le (n-1)(k-2) - (n-3)
 \]
 contradict each other.

Proposition~3 of \cite{Hoffmann21} states that if an $n$-state DFA $(Q,\Theta)$ has $m<n$ letters of defect~1 and every $(n-1)$-element subset of $Q$ is reachable, then the group of transformations on $Q$ generated by the letters of defect~0 has at most $m$ orbits. Applied to the DFA $\MCA(G,\alpha)$, where, as we have shown, the group is generated by $a$ and $m=1$, this implies that $a$ is a cyclic permutation of the vertices in $V$. Hence, the edges labeled $a$ form a Hamiltonian cycle.

\smallskip

\emph{Sufficiency}.  First observe that since the digraph $G$ is $k$-out-regular and, when passing to $G'$, we remove $k-2$ edges from $v$ to $x$ for each $v\in V$, the digraph $G'$ is $2$-out-regular.

Now let $C$ be a Hamiltonian cycle in the digraph $G$ and let $D$ be the set of edges of the subdigraph  $G'$ that do not belong to $C$. The vertex $x$ is the target of exactly one edge $e_x$ of $C$. If this edge is not among those removed when passing to $G'$, then every edge of $C$ is in $G'$. We label the edges of $C$ and $D$ by the letters $c$ and $d$, respectively. This road colors the digraph $G'$ into the DFA $(V,\{c,d\})$, in which $c$ acts as a cyclic permutation and $d$ has defect 1, since $V\setminus V{\cdot}d=\{x\}$. It follows from \cite[Theorem~1]{BV16} that a DFA with a prime number of states is completely reachable whenever it has a letter that acts as a cyclic permutation and a letter of defect~1. Therefore, the DFA $(V,\{c,d\})$ is completely reachable. Clearly, if we extend this coloring of $G'$ to a  coloring $\alpha$ of the digraph $G$, then the DFA $\MCA(G,\alpha)$ is also completely reachable.

Assume that the edge $e_x$ was removed when passing to $G'$. The other $n-1$ edges of $C$ remain in $G'$, whence $|D|=n+1$. Fix a vertex $z\ne x,y$. Then $\dd_{G'}^-(z) = 2$, that is, there are two edges in $G'$ with target $z$, one in $C$ and one in $D$. Let $e$ be the edge in $D$ with $t(e)=z$ and let $G''$ be the digraph with vertex set $V$ and edge set $C\cup(D\setminus\{e\})$.  We label the edges of $C$ and $D\setminus\{e\}$ by the letters $c$ and $d$, respectively. This road colors the digraph $G''$ into the DFA $(V,\{c,d\})$, in which $c$ acts as a cyclic permutation and $d$ has defect 1, since $V\setminus V{\cdot}d=\{z\}$. As above, \cite[Theorem~1]{BV16} implies that the DFA $(V,\{c,d\})$ is completely reachable, and this coloring of $G''$ extends to a completely reachable coloring of the digraph $G$.
\end{proof}

\begin{proof}[Proof of Theorem \ref{thm:NPC}]
Ples\'nik~\cite{Ple79} constructed, for an arbitrary Boolean formula $F$ in conjunctive normal form, a digraph $G =(V,E,s,t)$ such that
\begin{itemize}
\item[(i)] $|V|$ and $|E|$ are polynomial in the numbers of variables and clauses of $F$,
\item[(ii)] for every vertex $v\in V$, either $\dd_G^-(v)=1$ and $\dd_G^+(v)=2$, or $\dd_G^-(v)=2$ and $\dd_G^+(v)=1$,
\item[(iii)] every edge joins a vertex of out-degree $1$ to a vertex of in-degree $1$, or a vertex of out-degree $2$ to a vertex of in-degree $2$, and
\item[(iv)] $F$ is satisfiable if and only if $G$ has a Hamiltonian cycle.
\end{itemize}
The digraph $G$ in Ples\'nik's construction is also planar, but planarity is not needed for the argument below. Adding a parallel edge to each edge that joins a vertex of out-degree $1$ to a vertex of in-degree $1$ does not affect the presence or absence of a Hamiltonian cycle. Consequently, one may replace digraphs satisfying conditions (i)--(iv) by digraphs satisfying (i), (iv), and
\begin{itemize}
\item[(v)]  $\dd_G^-(v)=\dd_G^+(v)=2$ for every vertex $v\in V$.
\end{itemize}

In view of the NP-completeness of the Boolean satisfiability problem, it suffices to show that, for each $k\ge 2$, any digraph $G$ satisfying (v) can be transformed, in polynomial time, into a digraph to which Lemma~\ref{lem:Hami} applies.

We use the following operation. If $D = (V_D, E_D, s_D, t_D)$ and $H = (V_H, E_H, s_H, t_H)$ are digraphs with disjoint vertex and edge sets, $v$ is a vertex of $D$, and $x_1,x_2$ are distinct vertices in $H$, let  $\Replace(D,\frac{v}{x_1,x_2},H)$ denote the digraph with the vertex set $(V_D\setminus \{v\}) \cup V_H$, the edge set $E_D \cup E_H$, and the functions $s$ and $t$ defined by
\begin{align*}
  s(e)&: = \begin{cases}
    s_D(e) & \text{if $e \in E_D$ and $s_D(e) \neq v$},\\
    s_H(e) & \text{if $e \in E_H$},\\
    x_1 & \text{otherwise},
  \end{cases}\\
  t(e)&: = \begin{cases}
    t_D(e) & \text{if $e \in E_D$ and $t_D(e) \neq v$},\\
    t_H(e) & \text{if $e \in E_H$},\\
    x_2 & \text{otherwise.}
  \end{cases}
\end{align*}

A digraph to which Lemma~\ref{lem:Hami} applies must contain a subdigraph $G'$ with two vertices $x$ and $y$ having in-degrees $1$ and $3$, respectively, while every other vertex $z$ satisfies $\dd_{G'}^-(z)=2$. To achieve this, we use the \emph{funnel} digraph $F_3$ with three vertices $x,y,x_1$ and four edges: two from $x$, one to $y$ and one to $x_1$, and two from $y$, one to $x$ and one to $x_1$. Take an arbitrary digraph $G$ satisfying (v), choose any vertex $v_0$ of $G$, and define $G_1:= \Replace(G, \frac{v_0}{x_1,y},F_3)$. The transformation $G\rightsquigarrow G_1$ is illustrated in Fig.~\ref{fig:inserting F3}.

\medskip

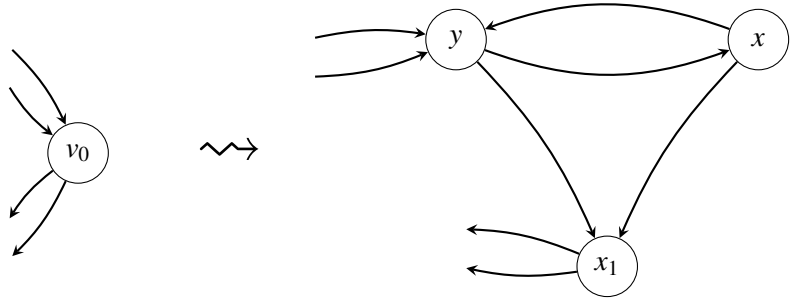
\begin{figure}[ht]
\centering
 \begin{tikzpicture}[xscale=-1,
  >=stealth,
  auto,
  state/.style={circle,draw,minimum size=8mm,inner sep=0pt}
]

\node[state] (x)  at (-2,2) {$x$};
\node[state] (y)  at ( 2,2) {$y$};
\node[state] (x1) at ( 0, -1) {$x_1$};

\path[->,thick] (x) edge[bend left=20]  (y);
\path[->,thick] (y) edge[bend left=20] (x);

\path[->,thick] (x)  edge[bend left=10] (x1);
\path[->,thick] (y) edge[bend right=10] (x1);

\node (pass) at (5,0.5)  {\Huge $\rightsquigarrow$};

\node[state] (v) at (7,0.5) {$v_0$};

\node (f1) at (8,-1) {};
\node (f2) at (8,-0.5) {};
\node (f3) at (8,1.5) {};
\node (f4) at (8,2) {};
\path[->,thick] (v)  edge[bend right=10] (f1);
\path[->,thick] (v)  edge[bend left=10]  (f2);
\path[->,thick] (f3) edge[bend left=10]  (v);
\path[->,thick] (f4) edge[bend right=10] (v);

\node (g1) at (2,-1) {};
\node (g2) at (2,-0.5) {};
\node (g3) at (4,1.5) {};
\node (g4) at (4,2) {};
\path[->,thick] (x1) edge[bend right=10] (g1);
\path[->,thick] (x1) edge[bend left=10] (g2);
\path[->,thick] (g3) edge[bend left=10] (y);
\path[->,thick] (g4) edge[bend right=10] (y);
\end{tikzpicture}
\caption{The transformation $G\rightsquigarrow\Replace(G, \frac{v_0}{x_1,y},F_3)$}\label{fig:inserting F3}
\end{figure}

Clearly, $\dd_{G_1}^-(x) = 1$, $\dd_{G_1}^-(y) = 3$, and $\dd_{G_1}^-(z) = 2$ for each $z \ne x,y$. It is easy to see that $G_1$ is 2-out-regular and has a Hamiltonian cycle if and only if so does $G$.

The next property we need to ensure to be able to apply Lemma \ref{lem:Hami} is the primality of the number of states. For this, for each positive integer $m$, consider the \emph{double line} digraph $L_m$ with vertex set $\{y_1, \dots, y_{m+1}\}$ that has two edges from $y_i$ to $y_{i+1}$ for all $1 \le i \le m$. Choose $m$ as the least positive integer such that $m + |V_{G_1}|$ is a prime number and let $G_2 := \Replace(G_1,\frac{x_1}{y_1,y_{m+1}},L_m)$. The transformation $G_1\rightsquigarrow G_2$ is illustrated in Fig.~\ref{fig:inserting Lm}.

\medskip

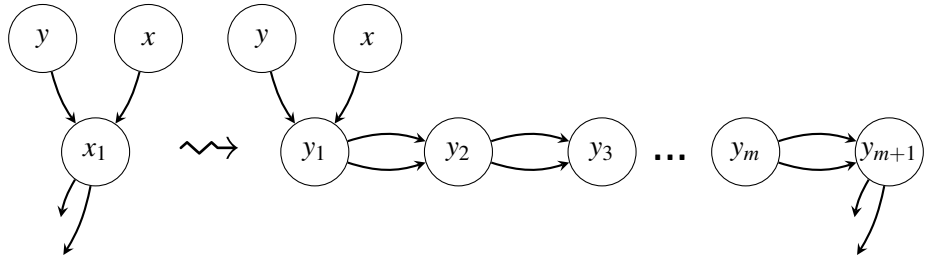
\begin{figure}[ht]
\centering
\begin{tikzpicture}[
  >=stealth,
  auto,
  state/.style={circle,draw,minimum size=9mm,inner sep=0pt}
]

\foreach \i in {1,...,3}
  \node[state] (y\i) at (1.9*\i,0) {$y_{\i}$};

\foreach \i in {1,...,2}{
  \pgfmathtruncatemacro{\j}{\i+1}
  \draw[->,thick,bend left=18]  (y\i) to (y\j);
  \draw[->,thick,bend right=18] (y\i) to (y\j);
}

\node (dots) at (6.6,0) {$\centerdot\centerdot\centerdot$};

\node[state] (ym) at (7.6,0) {$y_m$};
\node[state] (ym1) at (9.5,0) {$y_{m{+}1}$};
\draw[->,thick,bend left=18]  (ym) to (ym1);
\draw[->,thick,bend right=18] (ym) to (ym1);

\node[state] (v) at (-1,0) {$x_1$};

\node (f1) at (-1.5,-1.5) {};
\node (f2) at (-1.5,-1) {};
\node[state] (f3) at (-1.7,1.5) {$y$};
\node[state] (f4) at (-0.3,1.5) {$x$};
\path[->,thick] (v)  edge[bend left=10] (f1);
\path[->,thick] (v)  edge[bend right=10] (f2);
\path[->,thick] (f3) edge[bend right=10] (v);
\path[->,thick] (f4) edge[bend left=10]  (v);

\node (pass) at (0.5,0)  {\Huge $\rightsquigarrow$};

\node[state] (g3) at (1.2,1.5) {$y$};
\node[state] (g4) at (2.6,1.5) {$x$};
\path[->,thick] (g3) edge[bend right=10] (y1);
\path[->,thick] (g4) edge[bend left=10]  (y1);
\node (g1) at (9,-1.5) {};
\node (g2) at (9,-1) {};
\path[->,thick] (ym1)  edge[bend left=10] (g1);
\path[->,thick] (ym1)  edge[bend right=10] (g2);
\end{tikzpicture}
\caption{The transformation $G_1\rightsquigarrow\Replace(G_1, \frac{x_1}{y_1,y_{m+1}},L_m)$}\label{fig:inserting Lm}
\end{figure}

The digraph $G_2$ retains all properties of $G_1$---it is a 2-out-regular digraph such that $\dd_{G_2}^-(x) = 1$, $\dd_{G_2}^-(y) = 3$, $\dd_{G_2}^-(z) = 2$ for each $z \in V_{G_2}\setminus \{x,y\}$, and it has a Hamiltonian cycle if and only if so does $G_1$.  Besides, the number of states of $G_2$ is prime.

Define $G_3$ to be the digraph obtained from $G_2$ by adding $k-2$ edges from $v$ to $x$ for every vertex $v\in V_{G_2}$. Then $G_3$ has a Hamiltonian cycle if and only if $G_2$ has a Hamiltonian cycle and $G_3$ fulfills all conditions of Lemma \ref{lem:Hami} (with the subdigraph $G_2$ in the role of $G'$). Lemma \ref{lem:Hami} now implies that the digraph $G_3$ is completely reachable if and only if the initial digraph $G$ satisfying (v) has a Hamiltonian cycle. The described transformation $G\rightsquigarrow G_3$ requires time polynomial in the size of $G$ since the least prime greater than a given positive integer $n$ does not exceed $2n$ by Bertrand's postulate, and the sieve of Eratosthenes can be used to find this prime in $O(n\log\log n)$ time.
\end{proof}

\section{When all colorings are completely reachable} \label{sec:ARC}
For a vertex $u$ of a digraph $G$, an \emph{out-neighbor} of $u$ is any vertex $v$ such that $G$ has an edge from $u$ to $v$.

\begin{lemma} \label{lem:uniquesplit}
A digraph $G=(V,E,s,t)$ all of whose colorings are completely reachable has a unique vertex with more than one out-neighbor and contains a Hamiltonian cycle.
\end{lemma}

\begin{proof}
By Theorem~\ref{thm:CRC}, the digraph $G$ is strongly connected. Then every vertex has at least one out-neighbor. Assume that there exist two distinct vertices $u$ and $v$, each with at least two out-neighbors.

As explained in Sect.~\ref{subsec:simple}, we may assume that $G$ is simple. Consider the bipartite representation $BP(G)=(V',V'',L)$ of $G$. As shown in the proof of Theorem \ref{thm:CRC}, there is a matching $M\subset L$ that covers~$V'$. By the choice of $u$ and $v$, each of the vertices $u''$ and $v''$ is incident with at least two links of $L$. Hence there exist out-neighbors $u_1$ and $v_1$ of $u$ and $v$, respectively, such that the links $(u'_1,u'')$ and $(v'_1,v'')$ do not belong to $M$. (We do not claim that $u_1\ne v_1$; the argument below works equally well when $u_1=v_1$.) Using this observation, it is easy to see that for every link $\ell=(x_1,x_2)\in L\setminus M$, there exists a link $(y_1,y_2)\in L\setminus M$ such that $y_2\ne x_2$. Indeed, either $x_2\ne u''$ or $x_2\ne v''$; in the former case one may take $(u'_1,u'')$, and in the latter $(v'_1,v'')$. Now let $M_\ell$ denote the subset of $L$ obtained from $M$ by removing all links incident to $x_2$ and $y_2$ and adding the links $\ell=(x_1,x_2)$ and $(y_1,y_2)$.

Let $\mathcal{M} := \{M\} \cup \{M_\ell : \ell \in L\setminus M\}$. For every edge $e$ of the digraph $G$, write $\widehat{e}$ for the link $(t(e)', s(e)'')\in L$. Define a labeling $\alpha\colon E\to \mathcal{M}$ by
$\alpha(e): = \{N \in \mathcal{M} : \widehat{e} \in N\}$. To verify that $\alpha$ is a road coloring of $G$, consider, for an arbitrary vertex $w\in V$, the set $E_w$ of all edges with source $w$. For exactly one edge $e\in E_w$, one has $\widehat{e}\in M$; besides that, the link $\widehat{e}$ lies in each set $M_\ell$ such that $w''\ne x_2,y_2$, where $\ell=(x_1,x_2)$ and $(y_1,y_2)$ are the edges added in the construction of $M_\ell$. Every remaining set $M_\ell$ contains a link of the form $(z',w'')$, and hence belongs to the label of the unique edge in $E_w$ with target $z$. (The uniqueness follows from our assumption that the digraph $G$ is simple.) Thus, the family $\{\alpha(e): s(e)=w\}$ forms a partition of $\mathcal{M}$.

Let $\MCA: = \MCA(G,\alpha)$. Clearly, the letter $M$ acts as a permutation so its defect is 0. We claim that, for every link $\ell\in L\setminus M$, the defect of the letter $M_\ell$ is at least $2$. Indeed, let $\ell=(x_1,x_2)$, and let $(y_1,y_2)\notin M$ be the additional link used in the construction of $M_\ell$. Since the matching $M$ covers $V'$, there exist $z,t\in V''$ such that $(x_1,z),(y_1,t)\in M$. (The case $x_1=y_1$ is possible; in this case, of course, $z=t$.) Clearly, $z\ne x_2$ and $t\ne y_2$. Let $e_1,e_2,e_3,e_4$ be the preimages under the map $e\mapsto\widehat{e}$ of the links $(x_1,x_2),(y_1,y_2),(x_1,z)$, and $(y_1,t)$, respectively. Then $s(e_1){\cdot}M_\ell=s(e_3){\cdot}M_\ell$ and $s(e_2){\cdot}M_\ell=s(e_4){\cdot}M_\ell$ in $\MCA$. Therefore, $|V{\cdot}M_\ell|\le|V|-2$, that is, $\defect(M_\ell)\ge 2$.

Hence, the DFA $\MCA$ has no letter of defect $1$, and therefore, no $(|V|-1)$-element subset is reachable in $\MCA$. Thus, the digraph $G$ has a coloring that is not completely reachable, a contradiction.

Now let $w$ be the only vertex of $G$ with more than one out-neighbor. Since the digraph $G$ is strongly connected, $w$ belongs to every cycle of $G$. Call a cycle \emph{simple} if all its edges are distinct; in the rest of the proof the word `cycle' always means `simple cycle'. Every cycle is uniquely defined by an out-neighbor $u$ of $w$ as follows. If $u=w$, the corresponding cycle is simply the loop at $w$. If $u\ne w$, we start at $w$, first move to $u$, and then repeatedly move to the unique out-neighbor of the current vertex. The resulting path eventually returns to $w$, thus giving a cycle.

We order cycles by inclusion of their vertex sets. If a cycle $C$ is maximal with respect to this ordering, then the out-neighbor $u$ of $w$ defining $C$ must have $w$ as its only in-neighbor. Indeed, suppose that a vertex $v\ne w$ has an edge $e$ to $u$. Then $v$ is not a vertex in $C$ as otherwise the subpath of $C$ from $u$ to $v$, followed by $e$. would form a cycle not containing $w$, a contradiction. Since $G$ is strongly connected, there is a path from $w$ to $v$ in $G$; choose a shortest such path. This path, followed by $e$ and the subpath of $C$ from $v$ to $w$, forms a cycle in $G$ whose vertex set strictly contains that of $C$.

If a maximal cycle is unique, then it contains all vertices of $G$ and so is a Hamiltonian cycle. It remains to show that having more than one maximal cycle leads to a contradiction. Suppose that the cycles defined by two out-neighbors $u_1$ and $u_2$ of $w$ are both maximal. Consider the labeling $\beta$ of $G$ by $\dd_G^+(w)$ letters under which all edges with source $w$ get distinct labels and the label of each edge $w\to u$ is assigned to all edges of the cycle defined by $u$. Clearly, $\beta$ is a road coloring and so the DFA $\MCA(G,\beta)=(V,\Sigma)$ must be completely reachable. In particular, $\{u_1,u_2\}=V{\cdot}\sigma$ for some word $\sigma\in\Sigma^*$. Clearly, the word $\sigma$ is nonempty, and so $\sigma=\sigma'c$  for some $\sigma'\in\Sigma^*$ and $c\in\Sigma$. We have $(V{\cdot}\sigma'){\cdot}c=\{u_1,u_2\}$, which implies that the $V{\cdot}\sigma'$ contains in-neighbors of both $u_1$ and $u_2$. By the preceding paragraph, $w$ is the unique in-neighbor of each of $u_1$ and $u_2$. Thus, in $\MCA(G,\beta)$, we would have $w{\cdot}c=u_1$ and $w{\cdot}c=u_2$, which is impossible.
\end{proof}

Let $(\MBZ_n, \oplus)$ denote the additive group of integers modulo $n$. For each nonempty subset $S$ of $\MBZ_n\setminus\{0\}$, define $W(S,n)$ to be the digraph
\[
\bigl(\MBZ_n, \{(i,i\oplus 1) : i \in \MBZ_n\} \cup \{(n-1, s) : s \in S\}\bigr).
\]
The digraphs $W(S,n)$ may be viewed as slight generalizations of the well-known Wielandt digraph $W_n$ first defined (in the language of adjacency matrices) in~\cite[p. 648]{Wi50}; the digraph $W_n$  is precisely $W(\{1\},n)$; see Fig.~\ref{fig:wielandt}.
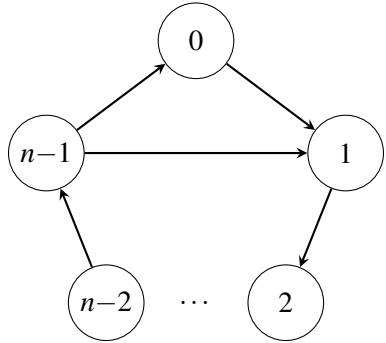
\begin{figure}[ht]
\centering
\begin{tikzpicture}[scale=0.99,
    >=stealth,
    vertex/.style={
        circle,
        draw,
        minimum size=10mm,
        inner sep=0pt
    },
    edge/.style={->,thick}
]

\node[vertex] (1) at (0,1.5) {$0$};
\node[vertex] (n) at (-2,0) {$n{-}1$};
\node[vertex] (2) at (2,0) {$1$};
\node[vertex] (nm1) at (-1.2,-2) {$n{-}2$};
\node[vertex] (3) at (1.2,-2) {$2$};

\draw[edge] (n) -- (1);
\draw[edge] (1) -- (2);
\draw[edge] (2) -- (3);
\draw[edge] (nm1) -- (n);
\draw[edge] (n) -- (2);

\node at (0,-2) {$\dots$};
\end{tikzpicture}
\caption{The Wielandt digraph $W_n=W(\{1\},n)$}
\label{fig:wielandt}
\end{figure}

It turns out that digraphs all of whose colorings are completely reachable fall in the family $W(S,n)$ up to simplification, that is, after gluing parallel edges into a single edge.

\begin{theorem}\label{thm:ACR}
 All colorings of a digraph $G$ with $n$ vertices are completely reachable if and only if the simplification of $G$ is isomorphic to the digraph $W(S,n)$ for a subset $S\subseteq\MBZ_n$ that generates the group $(\MBZ_n, \oplus)$.
\end{theorem}

\begin{proof}
 \emph{Necessity}. Lemma \ref{lem:uniquesplit} clearly implies that the simplification of the digraph $G$ is isomorphic to the digraph $W(S,n)$ for some subset $S$ of $\MBZ_n$. If the set $S$ generates a proper subgroup of the group $(\MBZ_n, \oplus)$, the greatest common divisor $d$ of the set $S\cup \{n\}$ is greater than 1. In this case, it is easy to see that $d$ divides the length of every cycle in $W(S,n)$ so that the period of $W(S,n)$ is not equal to 1, that is, $W(S,n)$ is not an aperiodic digraph. Then no completely reachable coloring of $W(S,n)$ exists by Theorem~\ref{thm:CRC}. Hence, $S$ must generate the entire group $(\MBZ_n, \oplus)$.

 \smallskip

\emph{Sufficiency}. Complete reachability of every coloring of each digraph $W(S,n)$ with the set $S$ generating the group $(\MBZ_n, \oplus)$ follows from \cite[Theorem~1]{BV16}.
\end{proof}

Clearly, the simplification $\bar{G}$ of a digraph $G=(V,E,s,t)$ can be constructed in $O(|E|)$ time. Checking whether $\bar{G}$ is isomorphic to a digraph $W(S,|V|)$ amounts to verifying that $\bar{G}$ is \scn{} and has a unique vertex of out-degree greater than~$1$, which can be done in $O(|V|+|E|)$ time. Finally, the subset $S\subseteq \MBZ_{|V|}$ generates the group $(\MBZ_{|V|}, \oplus)$ if and only if the greatest common divisor of the set $S\cup\{|V|\}$ is $1$, and this can be checked in $O(|V|)$ time using the Euclidean algorithm. Altogether, we see that Theorem~\ref{thm:ACR} yields an algorithm that determines whether all colorings of a given digraph $(V,E,s,t)$ are completely reachable in $O(|V|+|E|)$ time.

\smallskip

For comparison, we mention a recent result by D'Angeli and Rodaro~\cite{dAnRo26}, according to which the problem of determining whether all colorings of a given digraph are synchronizing is co-NP-complete.

\subsection*{Acknowledgments.} The authors thank the anonymous referees for their feedback and helpful suggestions, particularly for spotting a lacuna in the original proof of  Lemma \ref{lem:uniquesplit}.

\providecommand{\url}[1]{{#1}}
\providecommand{\urlprefix}{URL }
\expandafter\ifx\csname urlstyle\endcsname\relax
  \providecommand{\doi}[1]{DOI~\discretionary{}{}{}#1}\else
  \providecommand{\doi}{DOI~\discretionary{}{}{}\begingroup
  \urlstyle{rm}\Url}\fi

\end{document}